\documentclass[]{article}
\usepackage{graphics}
\usepackage{vaucanson-g}
\usepackage{xcolor}
\usepackage{verbatim}

\newtheorem{theorem}{Theorem}
\newtheorem{proposition}[theorem]{Proposition}
\newtheorem{corollary}[theorem]{Corollary}
\newtheorem{lemma}[theorem]{Lemma}
\newtheorem{remark}[theorem]{Remark}

\newtheorem{conjecture}{Conjecture}
\newtheorem{definition}{Definition}
\newenvironment{proof}{\noindent\textbf{Proof.}}{\hfill\rule{2mm}{2mm}\medskip}

\def\dd{\mathinner{\ldotp\ldotp}}       % dot dot

\newcommand{\RR}{\mathrm{Oroot}}
\newcommand{\LR}{\mathrm{Lroot}}
\newcommand{\rO}{\mathrm{\mbox{\scriptsize O}}}
\newcommand{\rr}{\mathrm{\mbox{\small r}}}
\newcommand{\rL}{\mathrm{\mbox{\scriptsize L}}}
\renewcommand{\fbox}[1]{#1}
\newcommand{\bigo}{{\mathcal O}}

\title{On the density of Lyndon roots in factors}
\author{%
Maxime Crochemore%
\thanks{King's College London %, Strand, London WC2R 2LS, UK,
 and Universit\'e Paris-Est.
 \texttt{Maxime.Crochemore@kcl.ac.uk}}
\and
Robert Merca\c{s}%
\thanks{Kiel University and King's College London.
 \texttt{robertmercas@gmail.com}}
}

\begin{document}

\maketitle

%---------%---------%---------%---------%---------%---------%---------%--------%
\section{Introduction}

The concept of a run %, i.e. maximal periodicity or maximal occurrence of repetitions,
 coined by Iliopoulos et al. \cite{Iliopoulos&Moore&Smyth:1997} when analysing
 repetitions in Fibonacci words, has been introduced to represent in a succinct manner
 all occurrences of repetitions in a word.
It is known that there are only $\bigo(n)$ many of them in a word of length $n$
from Kolpakov and Kucherov~\cite{Kolpakov&Kucherov:1999} who proved it in a
 non-constructive manner.
The first explicit bound was later on provided by Rytter~\cite{Rytter:2006}.
Several improvements on the upper bound can be found
 in~\cite{Rytter:2007,Crochemore&Ilie:2007b,PuglisiSS08:2008,Crochemore&Ilie&Tinta:2011, DezFra14}.
Kolpakov and Kucherov conjectured that this number is in fact smaller than $n$, which
 has been proved by Bannai et al.~\cite{BannaiIINTT14:2014,BaIInNaTaTs14}.
Recently, Holub~\cite{Hol15} and Fischer et al.~\cite{FisHolILew15}
 gave a tighter upper bound reaching $22n/23$.

In this note we provide a proof of the result, slightly different than the short
 and elegant proof in~\cite{BaIInNaTaTs14}.
Then we provide a relation between the border-free root conjugates of a square
 and the critical positions \cite[Chapter 8]{Lothaire:1997} occurring in it.
Finally, counting runs extends naturally to the question of their highest density,
 that is, to the question of the type of factors in which there is
 a large accumulation of runs.
This is treated in the last section.

\begin{figure}[ht]
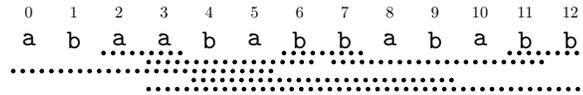

\begin{center}
%\ShowFrame\ShowGrid
\normalsize
\MediumPicture\VCDraw{%
\begin{VCPicture}{(13,5)(0,7)}
\ChgStateLineStyle{none}
% indices
\ChgStateLabelScale{.6}
\State[0]{(1,6.6)}{i0}
\State[1]{(2,6.6)}{i1}
\State[2]{(3,6.6)}{i2}
\State[3]{(4,6.6)}{i3}
\State[4]{(5,6.6)}{i4}
\State[5]{(6,6.6)}{i5}
\State[6]{(7,6.6)}{i6}
\State[7]{(8,6.6)}{i7}
\State[8]{(9,6.6)}{i8}
\State[9]{(10,6.6)}{i9}
\State[10]{(11,6.6)}{i10}
\State[11]{(12,6.6)}{i11}
\State[12]{(13,6.6)}{i12}
% letters
\ChgStateLabelScale{1}
\State[\tt a]{(1,6)}{c0}
\State[\tt b]{(2,6)}{c1}
\State[\tt a]{(3,6)}{c2}
\State[\tt a]{(4,6)}{c3}
\State[\tt b]{(5,6)}{c4}
\State[\tt a]{(6,6)}{c5}
\State[\tt b]{(7,6)}{c6}
\State[\tt b]{(8,6)}{c7}
\State[\tt a]{(9,6)}{c8}
\State[\tt b]{(10,6)}{c9}
\State[\tt a]{(11,6)}{c10}
\State[\tt b]{(12,6)}{c11}
\State[\tt b]{(13,6)}{c12}
% Lyndon roots
\RstEdge\ChgEdgeArrowStyle{-}
\ChgEdgeLineWidth{3}
\ChgEdgeLineStyle{dotted}
\Point{(2.6,5.7)}{r12} \Point{(3.4,5.7)}{r22}
\Point{(4.4,5.7)}{r33} \Edge{r12}{r33}{}
\Point{(3.6,5.5)}{s23} \Point{(5.4,5.5)}{s44}
\Point{(7.3,5.5)}{s66} \Edge{s23}{s66}{}
\Point{(2.6,5.3)}{t12} \Point{(5.4,5.3)}{t44}
\Point{(0.6,5.3)}{tt} \Point{(6.4,5.3)}{t55} \Edge{tt}{t55}{}
\Point{(6.6,5.7)}{r56} \Point{(7.4,5.7)}{r66}
\Point{(8.4,5.7)}{r77} \Edge{r56}{r77}{}
\Point{(5.6,5.1)}{u45} \Point{(8.4,5.1)}{u77}
\Point{(4.6,5.1)}{u34} \Point{(10.4,5.1)}{u99} \Edge{u34}{u99}{}
\Point{(3.6,4.9)}{v23} \Point{(8.4,4.9)}{v77}
\Point{(3.6,4.9)}{v23} \Point{(13.4,4.9)}{v1212} \Edge{v23}{v1212}{}
\Point{(8.6,5.5)}{s78} \Point{(10.4,5.5)}{s99}
\Point{(7.7,5.5)}{s67} \Point{(12.4,5.5)}{s1111} \Edge{s67}{s1111}{}
\Point{(11.6,5.7)}{r1011} \Point{(12.4,5.7)}{r1111}
\Point{(13.4,5.7)}{r1212} \Edge{r1011}{r1212}{}
\end{VCPicture}}
\end{center}
\caption{Dotted lines show the $8$ runs in $\texttt{abaababbababb}$.
For example, $[7\dd 11]$ is the run of period $2$ and length $5$ associated
with factor $\texttt{babab}$.}
\label{figu-runs}
\end{figure}

Formally, a \textit{run} in a word $w$ is an interval $[i\dd j]$ of positions,
 $0\leq i < j < |w|$, for which both the associated factor $w[i\dd j]$ is periodic
 (i.e. its smallest period $p$ satisfies $p \leq (j-i+1)/2$),
 and the periodicity cannot be extended to the right nor to the left:
 $w[i-1\dd j]$ and $w[i\dd j+1]$ have larger periods when these words are
 defined (see Figure~\ref{figu-runs}).
%Informally, when no confusion arises, at times we may also use the term ``run''
% for the factor $w[i\dd j]$ associated with the run.

%---------%---------%---------%---------%---------%---------%---------%--------%
\section{Fewer runs than length}

We consider an ordering $<$ on the word alphabet and the corresponding
 lexicographic ordering denoted $<$ as well.
We also consider the lexicographic ordering $\widetilde <$, called the reverse
 ordering, inferred by the inverse alphabet ordering $<^{-1}$.
The main element in the proof of the theorem is to assign to each run
 its greatest suffix according to one of the two orderings.

\begin{figure}[ht]
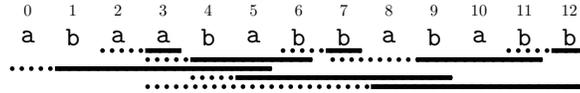

\begin{center}
%\ShowFrame\ShowGrid
\normalsize
\MediumPicture\VCDraw{%
\begin{VCPicture}{(13,5)(0,7)}
\ChgStateLineStyle{none}
% indices
\ChgStateLabelScale{.6}
\State[0]{(1,6.6)}{i0}
\State[1]{(2,6.6)}{i1}
\State[2]{(3,6.6)}{i2}
\State[3]{(4,6.6)}{i3}
\State[4]{(5,6.6)}{i4}
\State[5]{(6,6.6)}{i5}
\State[6]{(7,6.6)}{i6}
\State[7]{(8,6.6)}{i7}
\State[8]{(9,6.6)}{i8}
\State[9]{(10,6.6)}{i9}
\State[10]{(11,6.6)}{i10}
\State[11]{(12,6.6)}{i11}
\State[12]{(13,6.6)}{i12}
% letters
\ChgStateLabelScale{1}
\State[\tt a]{(1,6)}{c0}
\State[\tt b]{(2,6)}{c1}
\State[\tt a]{(3,6)}{c2}
\State[\tt a]{(4,6)}{c3}
\State[\tt b]{(5,6)}{c4}
\State[\tt a]{(6,6)}{c5}
\State[\tt b]{(7,6)}{c6}
\State[\tt b]{(8,6)}{c7}
\State[\tt a]{(9,6)}{c8}
\State[\tt b]{(10,6)}{c9}
\State[\tt a]{(11,6)}{c10}
\State[\tt b]{(12,6)}{c11}
\State[\tt b]{(13,6)}{c12}
% Lyndon roots
\RstEdge\ChgEdgeArrowStyle{-}
\ChgEdgeLineWidth{3}
\ChgEdgeLineStyle{dotted}
\Point{(2.6,5.7)}{r12} \Point{(3.4,5.7)}{r22}
\Point{(3.6,5.7)}{r23}
\Point{(4.4,5.7)}{r33} \Edge{r12}{r33}{}
\Point{(3.6,5.5)}{s23} \Point{(5.4,5.5)}{s44}
\Point{(4.6,5.5)}{s34}
\Point{(7.3,5.5)}{s66} \Edge{s23}{s66}{}
\Point{(2.6,5.3)}{t12} \Point{(5.4,5.3)}{t44}
\Point{(1.6,5.3)}{t01}
\Point{(0.6,5.3)}{tt} \Point{(6.4,5.3)}{t55} \Edge{tt}{t55}{}
\Point{(6.6,5.7)}{r56} \Point{(7.4,5.7)}{r66}
\Point{(7.6,5.7)}{r67}
\Point{(8.4,5.7)}{r77} \Edge{r56}{r77}{}
\Point{(5.6,5.1)}{u45} \Point{(8.4,5.1)}{u77}
\Point{(4.6,5.1)}{u34} \Point{(10.4,5.1)}{u99} \Edge{u34}{u99}{}
\Point{(3.6,4.9)}{v23} \Point{(8.4,4.9)}{v77}
\Point{(8.6,4.9)}{v78}
\Point{(3.6,4.9)}{v23} \Point{(13.4,4.9)}{v1212} \Edge{v23}{v1212}{}
\Point{(8.6,5.5)}{s78} \Point{(10.4,5.5)}{s99}
\Point{(9.6,5.5)}{s89}
\Point{(7.7,5.5)}{s67} \Point{(12.4,5.5)}{s1111} \Edge{s67}{s1111}{}
\Point{(11.6,5.7)}{r1011} \Point{(12.4,5.7)}{r1111}
\Point{(12.6,5.7)}{r1112}
\Point{(13.4,5.7)}{r1212} \Edge{r1011}{r1212}{}
\ChgEdgeLineStyle{solid}
\Edge{r23}{r33}{}
\Edge{r67}{r77}{}
\Edge{r1112}{r1212}{}
\Edge{s34}{s66}{}
\Edge{s89}{s1111}{}
\Edge{t01}{t55}{}
\Edge{u45}{u99}{}
\Edge{v78}{v1212}{}
\end{VCPicture}}
\end{center}
\caption{Plain lines show the $8$ greatest proper suffixes assigned to runs
 of $\texttt{abaababbababb}$ from Figure~\ref{figu-runs} in the proof of the
 theorem.
Note that no two suffixes start at the same position.}
\label{figu-mps}
\end{figure}

\begin{theorem}\label{theo-runs}
The number of runs in a word of length $n$ is less than $n$.
\end{theorem}

\begin{proof}
Let $w$ be a word of length $n$.
Let $[i\dd j]$ ($0\leq i < j < n$) be a run of smallest period $p$ in $w$.
If $j+1<n$ and $w[j+1] > w[j-p+1]$ we assign to the run the position $k$
 for which $w[k\dd j]$ is the greatest proper suffix of $w[i\dd j]$.
%Else, $w[j] < w[j-p]$ and the maximality is according to $\widetilde{<}$.
Else, $k$ is the position of the greatest proper suffix of $w[i\dd j]$
 according to $\widetilde{<}$.

Note that if $k>i$ then $k>0$, and that $w[k\dd j]$ contains a full period of
 the run factor, i.e. $j-k+1 \geq p$.
Also note that $w[k\dd k+p-1]$ is a greatest conjugate of the period root
 $w[i\dd i+p-1]$ according to one of the two orderings.
Therefore, it is border-free, known property of Lyndon words.

We claim that each position $k>0$ on $w$ is the starting position
 of at most one greatest proper suffix of a run factor.
Let us consider two distinct runs $[i\dd j]$ and $[\bar{i}\dd \bar{j}]$
 of respective periods $p$ and $q$, and which are called
 respectively the $p$-run and the $q$-run.
Assume $p \neq q$ since the runs cannot be distinct and have the same period.
For the sake of contradiction, we assume that their greatest suffixes
 share the same starting position $k$.

First case, $j = \bar{j}$, which implies $w[k\dd j]=w[k\dd \bar{j}]$.
Assume for example that $p<q$.
Then, $w[k\dd k+q-1]$ has period $p$ and thus is not border-free, which is
 a contradiction.

Second case, assume without loss of generality that $j < \bar{j}$ and
 that both suffixes are the greatest in their runs
 according to the same ordering, say $<$.
Let $d=w[j+1]$, the letter following the $p$-run.
By definition we have $w[j-p+1] < d$ and then $w[i\dd j-p+1] < w[i\dd j-p]d$.
But since $w[i+p\dd j]d$ is a factor of the $q$-run
 this contradicts the maximality of $w[k\dd \bar{j}-1]$.

Third case, $j \neq \bar{j}$ and the suffixes are greatest according
 to different orderings.
Assume without loss of generality that $p<q$ and the suffix of the $p$-run factor
 is greatest according to $<$.
Since $q>1$ we have both $w[k+q-1] \widetilde{>} w[k]$ and $w[k+q-1] = w[k-1]$,
 then $w[k-1] < w[k]$.
We cannot have $p>1$ because this implies $w[k-1] > w[k]$.
And we cannot have either $p=1$ because this implies $w[k-1] = w[k]$.
Therefore we get again a contradiction.

This ends the proof of the claim and shows that the number of runs is no more than
 the number $n-1$ of potential values for $k$, as stated.
\end{proof}

%---------%---------%---------%---------%---------%---------%---------%--------%
\section{Lyndon roots}

The proof of Theorem~\ref{theo-runs} by Bannai et al.~\cite{BaIInNaTaTs14}
 relies on the notion of a Lyndon root.
Recall that, for a fixed ordering on the alphabet, a Lyndon word is a primitive word
 that is not larger than any of its conjugates (rotations).
Equivalently, it is smaller than all its proper suffixes.
The root of a run $[i\dd j]$ of period $p$ in $w$ is the factor $w[i\dd i+p-1]$.
Henceforth, the Lyndon root of a run is the Lyndon conjugate of its root.
Therefore, since a run has length at least twice as long as its root,
 the first occurrence of its Lyndon root is followed by its first letter.
%Other Lyndon roots are obviously preceded by their last letter.
%a factor of the word associated with the run,
%factor that is a Lyndon word and whose length is the period of the run.
This notion of Lyndon root is the basis of the proof of the $0.5n$ upper bound
 on the number of cubic runs given in~\cite{CrochemoreIKRRW:2012}.
Recall that a run is said to be cubic if its length is at least three times
 larger than its period.

Lyndon roots considered in~\cite{BaIInNaTaTs14} are defined according to the two
 orderings $<$ and $\widetilde <$.
However, these Lyndon roots can be defined as smallest or greatest conjugates
 of the run root according to only one ordering.

The proof of Theorem~\ref{theo-runs} is inspired by the proof
 in~\cite{BaIInNaTaTs14} but does not use explicitly the notion of Lyndon roots.
The link between the two proofs is as follows: when the suffix $w[k\dd j]$
 is greatest according to $<$ in the run factor, then its prefix of period length,
 $w[k\dd k+p-1]$, is a Lyndon word according to $\widetilde <$.
As a consequence, the assignment of positions to runs is almost the same whatever
 greatest suffixes or Lyndon roots are considered.
%But since the run factor may contain more than one occurrence of the Lyndon roots,
% this leaves more flexibility to choose among them.

The use of Lyndon roots leaves more flexibility to assign positions
 to runs.
Indeed, a run factor may contain several occurrences of the run Lyndon root.
Furthermore, any two consecutive occurrences of this root do not overlap
 and are adjacent.
The multiplicity of these occurrences can be transposed to greatest
 suffixes by considering their borders.
Doing so, what is essential in the proof of Theorem~\ref{theo-runs} is that
 the suffixes and borders so defined are at least as long as
 the period of the run.
Consequently, consecutive such marked positions can be assigned to the same run.
As a consequence, since every cubic run is associated to at least two positions,
 this yields the following corollaries.

\begin{corollary}
If a word of length $n$ contains $c$ cubic runs, it contains less than $n-c$ runs.
\end{corollary}

\begin{corollary}
A word of length $n$ contains less than $0.5n$ cubic runs.
\end{corollary}

The last statement is proved in~\cite{CrochemoreIKRRW:2012} employing the notion
 of Critical position, which is discussed in the next section.

%---------%---------%---------%---------%---------%---------%---------%--------%
\section{Critical positions}

The consideration of the two above orderings appears in the
 simple proof of the Critical Factorisation Theorem
 \cite{CP91jacm} (for another proof see \cite[Chapter 8]{Lothaire:1997}).

Let us recall that the local period at position $|u|$ in $uv$ is the length
 of the shortest non-empty word $z$ for which $z^2$ is a repetition centred
 at position $|u|$.
Equivalently, in simpler words, $z$ is the shortest non-empty word that satisfies
 both conditions: either $z$ is a suffix of $u$ or $u$ is a suffix of $z$,
 and either $z$ is a prefix of $v$ or $v$ is a prefix of $z$.
Note that $vu$ satisfies the conditions but is not necessarily the shortest word
 to do it.
The Critical Factorisation Theorem states that a word $x$ of period $p$ admits a
 factorisation $x=uv$ whose local period at position $|u|$ is $p$.
Such a factorisation $uv$ of $x$ is called a critical factorisation and
 the position $|u|$ on $x$ a critical position.

When considering the starting positions of greatest suffixes defined above according
 to $<$ and to $\widetilde{<}$, the shorter of the two is known to provide
 a critical position following~\cite{CP91jacm}.
Thus, it does not come as a surprise to us that the simple proof of
 Theorem~\ref{theo-runs} relies on alphabet orderings.
Nevertheless, as the initial question does not involve any ordering on the alphabet,
 we could expect a proof using, for example, only the notion of critical positions.
The next lemma may be a step on this way.

\begin{figure}[ht]
\setlength{\unitlength}{1pt}
\fbox{\begin{picture}(321,50)
\put( 40,42){\framebox(120,6){}}\put(160,42){\framebox(120,6){}}
\put( 40,42){\makebox(50,6){$v$}}\put( 90,42){\makebox(70,6){$u$}}
\put(160,42){\makebox(50,6){$v$}}\put(210,42){\makebox(70,6){$u$}}
\put( 30,30){\framebox( 60,6){$y$}}\put( 90,30){\framebox( 60,6){$y$}}
\put(100,22){\framebox( 50,6){$v$}}\put(150,22){\framebox( 60,6){$\,\,\,q$}}
\put( 90, 10){\dashbox{1}(0,38){}}
\put( 90, 10){\framebox( 20,6){$p\,\,\,$}}\put(110, 10){\framebox( 70,6){$u$}}
\put(110,2){\framebox(100,6){$z$}}\put(210,2){\framebox(100,6){$z$}}
\put(210,8){\dashbox{1}(0,40){}}
\thicklines
\put(99, 9){\dashbox( 82,20){}}
\end{picture}}
\caption{%
If $uv$ is a border-free factor of $(vu)^2$, then at least one of its local period words
 $y$ or $z$ have length $|uv|$.
Otherwise, the common part in the dash-box has length equal to the sum of its periods $p$ and $q$ generating a contradiction.
}\label{figu-cpsq}
\end{figure}
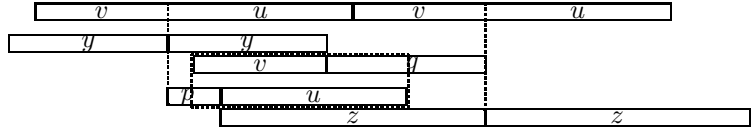

\begin{lemma}\label{lemm-cpsq}
Let $x^2 = (vu)^2$ be a square whose root conjugate $uv$ is border-free.
Then, at least $|v|$ or $|vuv|$ are critical positions on $x^2$.
\end{lemma}

\begin{proof}
Let $y$ be the local period word at position $|v|$ on $x^2$.
Since $uv$ is border-free, $v$ is a proper suffix of $y$.
Similarly, for the local period word $z$ at position $|vuv|$, the border-freeness
 of $uv$ implies that $u$ is a proper prefix of $z$.
The situation is displayed in Figure~\ref{figu-cpsq}.

For the sake of contradiction we assume the conclusion does not hold,
 i.e., both $y$ and $z$ are shorter than $uv$ (note that they cannot be longer than $uv$).
% his induces other occurrences of $v$ and of $u$ in $uv$, as shown in Figure~\ref{figu-cpsq}.

Let $|p|$ be the induced period of $pu$ and $|q|$ the induced period of $vq$.
The overlap between the two words $p$ and $q$ admits period lengths $|p|$
 and $|q|$ and has length $|pu|-(|uv|-|vq|)=|p|+|q|$.
Thus, by the Periodicity Lemma, $p$ and $q$ are powers of the same word $r$.
But then $r$ is a nonempty prefix of $u$ and a suffix of $v$ contradicting
 the border-freeness of $uv$.
\end{proof}

\paragraph{Example.}
Consider the square $\texttt{baba}$ of period $2$.
The occurrence of its border-free factor $\texttt{ab}$ induces the two critical
 positions $1$ and $3$.
On the contrary, the first occurrence of its border-free factor $\texttt{ba}$
 induces only one critical position, namely $2$, while the local period
 at $0$ has length $1<2$.

In the square $\texttt{abaaba}$ of period $3$, the occurrence of
 the border-free factor $\texttt{aab}$ produces the critical position $2$.
However, its position $5$ is not critical since the local period $2$
 is smaller than the whole period of the square.

%---------%---------%---------%---------%---------%---------%---------%--------%
\section{Lyndon roots density}

In this section we consider a generalisation of the problem of counting the
 maximal number of runs in a word.
In particular, we are interested in the following problem concerning
 first occurrences of Lyndon roots within a run factor.
Let us call the interval corresponding to the first such occurrence the $\LR$
 associated with the run.
Then, we are dealing with the following conjecture:

\begin{conjecture}[\cite{C14-cpm}]\label{pr:rep}
For any two positions $i$ and $j$ on a word $x$, $0\leq i \leq j < |x|$,
 the maximal number of run $\LR$s included in the interval $[i\dd j]$
 is not more than the interval length $j-i+1$.
\end{conjecture}

\begin{figure}[ht]
\begin{center}
%\ShowFrame\ShowGrid
\normalsize
\MediumPicture\VCDraw{%
\begin{VCPicture}{(14,2)(0,3)}
\ChgStateLineStyle{none}
% indices
\ChgStateLabelScale{.6}
\State[0]{(0,2.6)}{i0}
\State[1]{(1,2.6)}{i1}
\State[2]{(2,2.6)}{i2}
\State[3]{(3,2.6)}{i3}
\State[4]{(4,2.6)}{i4}
\State[5]{(5,2.6)}{i5}
\State[6]{(6,2.6)}{i6}
\State[7]{(7,2.6)}{i7}
\State[8]{(8,2.6)}{i8}
\State[9]{(9,2.6)}{i9}
\State[10]{(10,2.6)}{i10}
\State[11]{(11,2.6)}{i11}
\State[12]{(12,2.6)}{i12}
\State[13]{(13,2.6)}{i13}
\State[14]{(14,2.6)}{i14}
% letters
\ChgStateLabelScale{1}
\State[\tt a]{(0,2)}{c0}
\State[\tt b]{(1,2)}{c1}
\State[\tt a]{(2,2)}{c2}
\State[\tt b]{(3,2)}{c3}
\State[\tt a]{(4,2)}{c4}
\State[\tt a]{(5,2)}{c5}
\State[\tt b]{(6,2)}{c6}
\State[\tt a]{(7,2)}{c7}
\State[\tt b]{(8,2)}{c8}
\State[\tt b]{(9,2)}{c9}
\State[\tt a]{(10,2)}{c10}
\State[\tt b]{(11,2)}{c11}
\State[\tt a]{(12,2)}{c12}
\State[\tt b]{(13,2)}{c13}
\State[\tt b]{(14,2)}{c14}
% Lyndon roots
\RstEdge\ChgEdgeArrowStyle{-}
\ChgEdgeLineWidth{3}
\ChgEdgeLineStyle{dotted}
\Point{(3.6,1.7)}{r34}
\Point{(4.4,1.7)}{r43}
\Point{(4.6,1.7)}{r45}
\Point{(6.4,1.7)}{r66}
\Point{(3.6,1.5)}{s34}
\Point{(6.4,1.5)}{s66}
\Point{(6.6,1.5)}{s67}
\Point{(9.4,1.5)}{s99}
\Point{(7.6,1.7)}{r78}
\Point{(8.4,1.7)}{r88}
\Point{(4.6,1.3)}{t45}
\Point{(9.4,1.3)}{t99}
\ChgEdgeLineStyle{solid}
\Edge{r34}{r43}{}
\Edge{s34}{s66}{}
\Edge{r45}{r66}{}
\Edge{r78}{r88}{}
\Edge{s67}{s99}{}
\Edge{t45}{t99}{}
\end{VCPicture}}
\end{center}
\caption{Lines show the $6$ run $\LR$s inside the interval $[4\dd 9]$ corresponding to
 the factor $\texttt{aababb}$.}
\label{figu-LW}
\end{figure}

Let us consider the word
 $x=(\texttt{ab})^{k}\texttt{a}(\texttt{ab})^{k}\texttt{b}(\texttt{ab})^{k}\texttt{b}$
 and the interval of positions $[2k\dd 4k+1]$ corresponding to the factor
 $\texttt{a}(\texttt{ab})^{k}\texttt{b}$.
The number of $\LR$s corresponding to this interval is exactly the length
 $2(k+1)$ of the interval.
Figure~\ref{figu-LW} shows the situation when $k=2$.
This example gives a lower bound on the maximal number of $\LR$s contained
 in an interval of positions.

\begin{proposition}\label{prop:lower}
The number of $\LR$s contained in an interval of positions on a word, can be
 as large as the length of the interval.
%For $x=(\texttt{ab})^{k}\texttt{a}(\texttt{ab})^{k}\texttt{b}(\texttt{ab})^{k}\texttt{b}$, the factor $w=\texttt{a}(\texttt{ab})^{k}\texttt{b}$ contains $|w|$ $\LR$s.
\end{proposition}

In addition to the conjecture, we believe that factors associated with
 intervals of length at least $4$ containing the maximal number of $\LR$s are
 of the form $a(ab)^+b$ for two different letters $a$ and $b$.
It can be checked that the maximal number of $\LR$s is respectively $1$ and $3$
 for intervals of lengths $1$ and $3$ with factors $a$ and $aab$,
 but is only $1$ for intervals of length $2$.
All these factors are Lyndon words for the ordering $a<b$.
This is due to the fact that such factors contain overlapping Lyndon roots
 making the whole factor a Lyndon word itself.

%This provides us with a lower bound for our problem.
\begin{remark}\label{rem:interv}
In order to obtain an upper bound on the number of $\LR$s inside an interval
of positions, it is enough to restrict ourselves to counting the maximal
number of $\LR$s within an interval corresponding to some Lyndon word.
\end{remark}

Indeed, each $\LR$ corresponds to a Lyndon word. Since we want an interval
that contains the maximal such number, all the positions of this interval are
covered by some Lyndon word. However, since the overlap between every two
Lyndon words produces a Lyndon word, and since every word can be expressed as
a concatenation of Lyndon words, our claim follows.

We show that the number of $\LR$s inside an interval corresponding to a
Lyndon word is bounded by $1.5$ times the length of the interval. For this we
make use of the result from~\cite{BaIInNaTaTs14} stating that each position
of a word is the starting position of at most one specific root associated
with the run. The root is chosen according to some order defined by the
letter following the run. We denote such a root relative to the order as the
$\RR$ of the run. Formally:
%stating that with each run we can associate a root that is chosen according to some order defined by the letter following the run, and that at each position within the word there is at most one such root starting. We denote such a root relative to the order as the $\RR$ of the run. Formally:

\begin{definition}
Let $r$ be a run of period $p_r$ of the word $w$ and let $r_\rL$ be the $\LR$
associated with $r$. If $r$ ends at the last position of $w$, or if the
letter at the position following $r$ is smaller than the letter $p_r$
positions before it, then the $\RR$ is the interval corresponding to the
first occurrence of a Lyndon root that is not a prefix of $r$. Otherwise, the
$\RR$ is the interval corresponding to the length $p_r$ prefix of the
greatest proper suffix of the run factor $r$.
% (the maximal suffix different from $r$ itself).
%If the letter following $r$ is smaller than the letter $p_r$ positions before, or $r$ is a suffix of $w$
\end{definition}

Observe that since a run is at least as long as twice its minimal period,
this ensures the existence of both its $\LR$ as well as its $\RR$. To see
that the $\RR$ is never a prefix of the run factor it is associated with,
observe from its above definition that in the second case this is actually
the interval corresponding to the length $p_r$ prefix of the maximal proper
suffix of the run, which is different from the run itself (being proper).

Henceforth we fix an interval $[i\dd j]$ with its corresponding Lyndon word
$w$ of length $\ell$.
%Let us denote the factor preceding $w$ as $w_{\rm p}$, while the one following it by $w_{\rm s}$, such that the interval corresponding to $w_{\rm p}ww_{\rm s}$ is the minimal interval that contains all runs that have their $\LR$s in $[i\dd j]$.
Furthermore, we denote by $r_\rL = [i_\rL\dd j_\rL]$ the $\LR$ of the run $r
= [i_r\dd j_r]$ and by $r_\rO= [i_\rO \dd j_\rO]$ its $\RR$. For $r$, we
denote by $p_\rr$ the (smallest) period of the run. Please note that
$|r_\rL|=|r_\rO|=p_r$, while both must start and end within the run $r$.

We make the following remarks based on the already known properties of $\LR$s
and $\RR$s.

\begin{remark}\label{rem:roots}
The $\LR$  and the $\RR$ associated with a run $r$ start within the first
$p_r$ and $p_r+1$, respectively, positions of the run, and both have length
$p_r$.
\end{remark}

%Furthermore, since the $\RR$ corresponds to a minimal or maximal conjugate of the period of the run, we know from~\cite{BaIInNaTaTs14} the following results:

%The following observation is a direct consequence of the definition of the $\RR$:
As a direct consequence of the definition of the $\RR$ we have the following:
\begin{remark}\label{lem:i_L-before}
If the $\RR$ of a run $r$ is a Lyndon word, then the $\RR$ and the
corresponding $\LR$ represent the same factor and, either $i_\rO=i_\rL$, or
the run $r$ starts at position $i_\rL$ and $i_\rL+p_r  = i_\rO$.
\end{remark}
%\begin{proof}
%The statement says that if the $\RR$ is a Lyndon word, then it equals the $\LR$. Furthermore, according to the definition of a $\RR$, these either start at the same position when the run does have $r_L$ as a prefix, or the $\LR$ starts at the beginning of the run, while its $\RR$ immediately after the $\LR$, and both of them are equal.
%%The result is straightforward as both the $\LR$ and the $\RR$ occur within the first $p_r+1$ positions of the run. The only time the $\RR$ starts on position $i_L+p_r$ is precisely when it is equal to the $\LR$ associated to the same run.
%\end{proof}

In conclusion we have the following:

\begin{remark}%\label{rem:roots}
To bound the number of $\LR$s inside the interval $[i\dd j]$ corresponding to
the word $w$, it is enough to consider all runs starting within the interval
corresponding to a factor $w_{\rm p}w$, where $|w_{\rm p}| < |w|$.
% and $|w_{\rm s}| \leq |w|$.
\end{remark}

For the rest of this work let us fix the factor preceding a Lyndon word $w$
as $w_{\rm p}$, while the one following it by $w_{\rm s}$, such that the
interval corresponding to $w_{\rm p}ww_{\rm s}$ is the shortest interval that
contains all runs with their $\LR$s in $[i\dd j]$.

Now we start looking at the relative positions of the $\RR$ and $\LR$
corresponding to the same run.

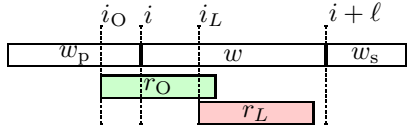
\begin{figure}[ht]
\setlength{\unitlength}{1pt}
\setlength\fboxsep{0pt}
\center
%\begin{minipage}{.35\textwidth}
%\begin{picture}(00,40)
%\put(72,10){\colorbox{green!20}{\framebox( 43,8){$r_R$}}}
%\put(0,22){\framebox(50,8){$w_{\rm p}$}}\put(50,22){\framebox(70,8){$w$}}\put(120,22){\framebox(30,8){$w_{\rm s}$}}
%\put(35,0){\colorbox{red!20}{\framebox( 43,8){$$}}}
%\put(55,10){\makebox(0,0)[lb]{$r_L$}}
%\put(40,2){\makebox(0,0)[lb]{$x$}}
%\put(55,0){\makebox(0,0)[lb]{$y$}}
%\put(73,2){\makebox(0,0)[lb]{$z$}}
%\put(35,0){\dashbox{1}(0,37){}}
%\put(50,0){\dashbox{1}(0,37){}}
%\put(72,0){\dashbox{1}(0,37){}}
%\put(120,0){\dashbox{1}(0,37){}}
%\put(35,37){\makebox(0,0)[lb]{$i_L$}}
%\put(72,37){\makebox(0,0)[lb]{$i_{R}$}}
%\put(51,38){\makebox(0,0)[lb]{$i$}}
%\put(121,38){\makebox(0,0)[lb]{$i+\ell$}}
%\end{picture}
%\caption{$r_L$ before Lyndon word
%}\label{fig:BR_after}
%\end{minipage}%
%\begin{minipage}{.08\textwidth}
% {
% \hfill
% }
%\end{minipage}%
%\begin{minipage}{.5\textwidth}
\begin{picture}(150,40)
\put(72,0){\colorbox{red!20}{\framebox( 43,8){$r_L$}}}
\put(0,22){\framebox(50,8){$w_{\rm p}$}}\put(50,22){\framebox(70,8){$w$}}\put(120,22){\framebox(30,8){$w_{\rm s}$}}
\put(35,10){\colorbox{green!20}{\framebox( 43,8){$r_\rO$}}}
\put(35,0){\dashbox{1}(0,37){}}
\put(50,0){\dashbox{1}(0,37){}}
\put(72,0){\dashbox{1}(0,37){}}
\put(120,0){\dashbox{1}(0,37){}}
\put(35,37){\makebox(0,0)[lb]{$i_\rO$}}
\put(72,37){\makebox(0,0)[lb]{$i_{L}$}}
\put(51,38){\makebox(0,0)[lb]{$i$}}
\put(121,38){\makebox(0,0)[lb]{$i+\ell$}}
\end{picture}
\caption{$r_\rO$ starts before the Lyndon word $w$
}\label{fig:BR_before}
%\end{minipage}
%\begin{minipage}{.15\textwidth}
% {
% \hfill
% }
%\end{minipage}%
\end{figure}

%\begin{lemma}\label{lem:LR>BR}
%If the $\RR$ of a run is a prefix of a Lyndon word, then the run starts at position $i_L$.
%\end{lemma}
%\begin{proof}
%Assume this is the case. Then it must be that the $\RR$ is chosen according to the lexicographical ordering and $\LR=\RR$. Hence, according to Lemma~\ref{lem:bannai}, it cannot be that the $\LR$ comes after the $\RR$, and we can again consider the case illustrated in Figure~\ref{fig:BR_after}, where $y$ is this time the empty word. In this case, if $z\neq\varepsilon$, since $r_L=xz$ is an $\LR$
%and $r_R=zx$, we have that $xz=zx$. However this implies that both $x$ and $z$ are powers of the same word, which is a contradiction.
%The conclusion follows.
%\end{proof}

\begin{lemma}\label{lem:BR>LR}
%Consider an $\LR$ which occurs within an interval $w=[i\dd i+\ell-1]$. If the
Fix an $\LR$ occurring within an interval $w=[i\dd i+\ell-1]$ of a word. If the
$\RR$ corresponding to the same run starts outside $w$, then either:
 \begin{enumerate}
 \item\label{lem:BR_after}
%  the $\LR$ ends on the same position as the Lyndon word which is equal to the $\RR$ starting at the position immediately after the end of the Lyndon word, or \par
 the $\LR$ ends at position $i+\ell-1$, and the $\RR$ starts at position
 $i+\ell$,~or %\par
 \item\label{lem:BR_before} the $\LR$ starts at position $i$, and the $\RR$
 starts before position $i$. %\par
 \end{enumerate}
\end{lemma}

\begin{proof}
%Let us denote by $w$ the Lyndon word, by $w_{\rm p}$ the factor preceding it and by  $w_{\rm s}$ the one following $w$. Furthermore, we
Following Remark~\ref{rem:interv}, without loss of generality assume that
$w_{\rm p}$ starts at position 0 and has length $i$, while $w=[i\dd
i+\ell-1]$ is a Lyndon word.
%We consider a run $r$ that has $r_L$ as $\LR$ occurring at position $i_L$ and $r_R$ as $\RR$ occurring at position $i_R$.
As stated in the hypothesis, $i\leq i_\rL$ and $|r_\rL| \leq \ell$.\par

First let us assume that $i_\rO \geq i+\ell$, hence the $\RR$ starts after
the end of the Lyndon word $w$. The result follows immediately from
Remark~\ref{lem:i_L-before}.\par

For the second statement, consider Figure~\ref{fig:BR_before} where
$i_\rO<i\leq i_\rL$. Assume towards a contradiction that $i < i_\rL$. Since
$i_\rO<i$ and $r_\rL$ corresponds to an interval on $w$, it must be that the
corresponding run starts before or on position $i_\rO$ and it ends after or
on position $i+\ell-1$. However, since $w$ is a Lyndon word, it must be that
for any word $x$ such that $y x$ is a suffix of $w$, where $y$ is the factor
corresponding to $r_\rL$, we have $w < y x$. But in this case, unless $y$ is
a prefix of $w$, we get a contradiction with the fact that $r_L$ is a Lyndon
root. In the former case, however, we get a contradiction with the definition
of $r_\rL$ since $i < i_\rL$ ($r_\rL$ is the interval corresponding to the
first occurrence of a Lyndon root of the run), and the conclusion follows in
this case as well.
\end{proof}

%It is well known that every word can be expressed as a concatenation of max
As a consequence of the above lemma, the number of $\LR$s is bounded by
$2\ell$. This is because the $\RR$s all start inside an interval of length
$|w_{\rm p}|+|w|+1\leq 2\ell$, and no two share the same starting
position~\cite{BaIInNaTaTs14}. In the following we reduce this bound to
$1.5\ell$. The next lemma shows that the situation in
Lemma~\ref{lem:BR>LR}.\ref{lem:BR_after} is met by at most one $\RR$.

\begin{lemma}\label{lem:1BR_after}
%For any word and any of its factors that are Lyndon words, there exists at most one run that has its $\RR$ starting after the end of the Lyndon word while its $\LR$ corresponds to an interval inside the Lyndon word.
For any word and any interval $[i\dd j]$ on it, there exists at most one run
that has its $\RR$ starting after position $j$ while its $\LR$ corresponds to
an interval inside $[i\dd j]$.
%For every Lyndon word, there is at most one $\RR$ that has its starting position after its end and has its associated $\LR$ corresponding to a factor of the Lyndon word.
\end{lemma}

\begin{proof}
According to Lemma~\ref{lem:BR>LR}.\ref{lem:BR_after}, it must be the case
that $i_\rO=j+1$, while $i_\rL=j-p_r+1$, for any appropriate run $r$.
However, having more than one $\RR$ starting at position $j+1$ with the
factor corresponding to its $\LR$ as a suffix of the Lyndon word, would then
imply that the larger of $\LR$s that corresponds to a Lyndon word is
bordered, which is a contradiction.
%The conclusion follows.
\end{proof}

Now we are dealing with $\RR$s corresponding to
Lemma~\ref{lem:BR>LR}.\ref{lem:BR_before}.
%
%
%As a consequence of Lemma~\ref{lem:BR>LR}.\ref{lem:BR_before}, we note that there are at most $\frac{\ell-1}{2}$ runs that start in $w_{\rm p}$ such that their corresponding $\RR$s start in $w_{\rm p}$, while their $\LR$s start in $w$ (and even more, the $\LR$s are prefixes of $w$). Hence, combining this with the fact that within $w$ we have at most $\ell$ $\RR$s starting there, see~\cite{BaIInNaTaTs14}, and there is at most one $\RR$ that can have the $\LR$ in $w$, according to Lemma~\ref{lem:1BR_after} we can now give an upper bound for our problem:

\begin{proposition}
For a given word, any interval of length $\ell$ of positions on the word contains at
most $3\ell/2$ $\LR$s.
\end{proposition}

%\begin{proof}
%Since for any $\RR$ starting in the interval associated with $w_{\rm p}$, the only $\LR$s corresponding to them in $[i\dd i+\ell-1]$ start at position $i$, we note that these can be bounded by the number of times the first letter of $w$ is preceded by a different symbol in $w$. Since this number is obviously bounded by $\frac{\ell}{2}$, while the whole length of $w_{\rm p}$ is bounded by $\ell-1$, we conclude that there are less than $\frac{\ell-1}{2}$ runs that start in $w_{\rm p}$ such that their corresponding $\RR$s start in $[i-|w_{\rm p}|\dd i]$, while their $\LR$s start in $w$ (the $\LR$s correspond to prefixes of $w$). Hence, combining this with the fact that within $w$ we have at most $\ell$ $\RR$s starting there, see~\cite{BaIInNaTaTs14}, and since according to Lemma~\ref{lem:1BR_after} there is at most one $\RR$ starting after position $j$ that has the $\LR$ in $w$,  we get an upper bound for our problem.
%\end{proof}
\begin{proof}
Let us denote once more our interval by $w=[i\dd i+\ell-1]$ and the interval
preceding it by $w_{\rm p}$.
We know from the definition that an $\LR$ is the first Lyndon root of a run,
and therefore the letter ending every $\LR$ must be greater than its first
letter, while the one on the position right after the end of the $\LR$ must
be the same as the first letter of the $\LR$.

Since for any $\RR$ starting in the interval associated with $w_{\rm p}$, the
$\LR$ corresponding to it in $w$ starts at position $i$, we note that these
can be bounded by the number of length two factors in $w$ that have the
letter on the first position larger than the letter on position $i$, while
the second one identical. If the second letter of such a factor is smaller
than the letter on position $i$, than this situation would make it impossible
for a $\LR$ to start on position $i$ and end before this position (this is
because an $\LR$ is the first Lyndon root occurrence of a run).

Since this number is obviously bounded by $\frac{\ell}{2}$, while the whole
length of $w_{\rm p}$ is bounded by $\ell-1$, (by considering the symmetric
situation) we conclude that there are less than $\frac{\ell-1}{2}$ runs that
start in $w_{\rm p}$ such that their corresponding $\RR$s start before
position $i$, while their $\LR$s start at position $i$ (the $\LR$s correspond
to prefixes of $w$). Hence, combining this with the fact that within $w$ we
have at most $\ell$ $\RR$s starting there, see~\cite{BaIInNaTaTs14}, and
since according to Lemma~\ref{lem:1BR_after} there is at most one $\RR$
starting after position $j$ that has the $\LR$ in $w$,  we get an upper bound
for our problem.
\end{proof}

The bound given in the above proposition is not really tight. On this point
let us complete the conjecture:

\begin{conjecture} %[\cite{C14-cpm}]
For a given word, any interval of length $\ell>0$ of positions on the word contains
at most $\ell$ $\LR$s, and the maximum number is obtained only
when the factor corresponding to the interval is of the form
$a(ab)^{\frac{\ell-2}{2}}b$, where $\ell>3$, and the letters $a$ and $b$ satisfy $a<b$.
\end{conjecture}

We end this article with a few more observations regarding the results
from~\cite{BaIInNaTaTs14}, when we restrict ourselves to binary words. First
we recall a property of $\RR$s:
%
%\begin{lemma}\label{lem:LW-allign}
%Two maximal extending Lyndon words overlap if and only if one is a factor of the other.
%\end{lemma}
%\begin{proof}
%Assume that this is not the case and we have a word $y$ as the overlap of two different maximal Lyndon words, $xy$ and $yz$, respectively. Since $xy$ is maximal, it follows that $xyz$ must contain a border. It is straightforward that such a border must have length greater than $y$ and be shorter than half the length of the whole word.
%\end{proof}

\begin{lemma}[Bannai et al.~\cite{BaIInNaTaTs14}]\label{lem:BR-overlaps}
If two different $\RR$s obtained considering the same order overlap, then
their overlap is the shortest of the $\RR$s.
\end{lemma}
%\begin{proof}
%Following Lemma~\ref{lem:LW-allign}, since the two $\RR$s are obtained according to the same order, the result is available.
%\end{proof}

We observe that we can consider $\RR$s to be obtained according to a certain
order based on the letter that these $\RR$s start with (thus all $\RR$s
starting with $\texttt{a}$ are obtained according to the lexicographical
order, while the ones starting with $\texttt{b}$ are obtained according to
the inverse lexicographical one).

\begin{proposition}\label{lem:bound_order}
For a given binary word, any interval of length $\ell$ of positions on the word
contains at most $\frac{\ell-1}{2}$ $\RR$s obtained according to the same
order.
%In a length $\ell$ interval corresponding to a binary factor there are at most $\frac{\ell-1}{2}$ runs whose $\RR$s are obtained according to the same order.
\end{proposition}

\begin{proof}
Without loss of generality we fix an order; let us say lexicographical.
Observe first, that for a word to correspond to an $\RR$, whenever they are
not binary, they must start with a letter $\texttt{a}$ and end with a
$\texttt{b}$ (as previously mentioned).  Furthermore, it must be the case
that this interval is preceded by a $\texttt{b}$ and followed by an
$\texttt{a}$, as otherwise it does not correspond to a Lyndon root (there
exists another rotation that has an extra $\texttt{a}$ in its longest unary
prefix).

Finally, observe that considering their relative position, following
Lemma~\ref{lem:BR-overlaps}, two such $\RR$s are either included one in the
other, or they are disjoint.

Now, considering two words corresponding to two $\RR$s, let us say $u$ and
$v$ with $u$ a factor of $v$, we note that, since their lengths are
different, following the initial conditions, they must differ by a length of
at least $2$, whenever $u$ is not unary (each starts between a $\texttt{b}$
and an $\texttt{a}$, and ends between an $\texttt{a}$ and a $\texttt{b}$).
For the unary case, note that every block of consecutive $\texttt{a}$'s must
be in-between two occurrences of $\texttt{b}$. Furthermore, we cannot have
two unary words corresponding to $\RR$s overlapping each other. Thus if the
position of the second $a$ is an $\RR$ in the word
$\texttt{ba}^{\ell}\texttt{b}$, for $\ell>0$, it is impossible to have a
length less than $3$ for any word starting with the first $a$ whose interval
corresponds to an $\RR$.

Given that for any two distinct adjoining $\RR$s both their lengths and the
number of $\RR$s they contain add up, the result follows in this case as
well.

In order to get the $-1$, we observe that for any word of length at least
$3$, for the interval it determines to have the maximum number of $\RR$s of
the same order, according to the previous facts, would imply the word to have
the form $(\texttt{ab})^+$. However, now, the $\RR$s would correspond to
words that are just powers of one another, contradicting their property of
being Lyndon words.
\end{proof}

%{\bf CONTINUE FROM HERE!!!!! MAKE ALL REMARKS}
\begin{comment}
\begin{remark}\label{lem:bound_order_factor}
For a given order, every length $\ell$ factor of a binary word completely contains at most $\frac{\ell-1}{2}$ $\RR$s of runs corresponding to the word that are obtained from.
\end{remark}
\begin{proof}
The proof follows the lines of the above one, with the only mention that here, we might get an extra 1 for the number of the $\RR$s, as there might also be one starting at the beginning of the factor (the run might start before the beginning of the factor, thus not restricting us).
\end{proof}
\end{comment}

Furthermore, denoting by $|w|_u$ the number of all (possibly overlapping)
occurrences of $u$ in $w$, as consequence of the above we have the following:

\begin{corollary}\label{cor:strict_bound_order}
Every length $\ell$ interval associated with a factor $w$ of a binary word
completely contains at most $\min\{|w|_{ab},|w|_{ba}\}$ $\RR$s that
correspond to non-unary factors and are obtained according to the same order.
\end{corollary}
%\begin{proof}
%Let us again consider the lexicographical order, the other case being symmetrical.
%As mentioned before, all non-unary $\RR$s start within a $\texttt{ba}$ and end within an $\texttt{ab}$. Furthermore, whenever they overlap they must be included one in the other.
%
%Since between every pair of $ab$'s there is at least one $\texttt{ba}$ and vice-versa, the result follows.
%\end{proof}

\begin{corollary}\label{cor:unary_bound_order}
The number of $\RR$s associated with unary runs within every factor of a
binary word is at most one extra than the number of unary maximal blocks
within the factor (by a maximal block we refer to a unary factor that cannot
be extended either to the left or to the right without losing its
periodicity).
\end{corollary}
%\begin{proof}
%Obviously for each unary maximal block we have precisely one unary $\RR$ situated on the second letter of the block. Furthermore, we might get an extra $\RR$ situated on the first position of the factor, when the factor is preceded by the same letter it starts with.
%\end{proof}

%%%%%%%%%%%%%%%%%%%%%%%%%%%%%%%%%%%%%%%%
%%%%%%%%%%%%%%%%%%%%%%%%%%%%%%%%%%%%%%%%
%%%%%%%%%%%%%%%%%%%%%%%%%%%%%%%%%%%%%%%%
%%%%%%%%%%%%%%%%%%%%%%%%%%%%%%%%%%%%%%%%

%---------%---------%---------%---------%---------%---------%---------%--------%
\section{Acknowledgement}

We would like to warmly thank Gregory Kucherov, Hideo Bannai, and Bill Smyth
for helpful discussions on the subject. The work of Robert Merca\c{s} was
supported by the P.R.I.M.E. programme of  DAAD with funds provided by the
Federal Ministry of Education and Research (BMBF) and the European Union's
Seventh Framework Programme for research, technological development and
demonstration (grant agreement no. 605728).

%---------%---------%---------%---------%---------%---------%---------%--------%
\bibliographystyle{abbrv}
\bibliography{runs_Lroots}
%\bibliography{MRbib}
%---------%---------%---------%---------%---------%---------%---------%--------%
\end{document}